\documentclass[11pt]{article}

\usepackage[left=2.5cm,right=2.5cm,top=2cm,bottom=2cm,includeheadfoot]{geometry}
\usepackage{caption}
\usepackage{dsfont}
\usepackage{amsmath, amsthm, amssymb, amstext, amsfonts}
\usepackage{mathtools}
\usepackage{nicefrac}
\usepackage{enumerate}
\usepackage{graphicx}
\usepackage{xcolor}
\usepackage{hyperref}
\usepackage{url}
\usepackage[normalem]{ulem}

\usepackage[backend=biber, style=alphabetic, backref=true, hyperref=true]{biblatex}
\usepackage[babel,english=british]{csquotes}
\DefineBibliographyStrings{english}{%
    backrefpage  = {cited on p.}, 
    backrefpages = {cited on pp.} 
}

\renewcommand{\L}{\mathcal{L}}

\newcommand{\IR}{\mathbb{R}}
\newcommand{\IN}{\mathbb{N}}

\newcommand{\IC}{\mathbb{C}}

\newcommand{\norm}[1]{\left\lVert#1\right\rVert}

\numberwithin{equation}{section}

\newtheorem{thm}{Theorem}[section]
\newtheorem*{thm*}{Theorem}

\newtheorem{prp}[thm]{Proposition}

\newtheorem{crl}[thm]{Corollary}
\newtheorem{cnj}[thm]{Conjecture}
\newtheorem{lmm}[thm]{Lemma}
\newtheorem{dff}[thm]{Definition}

\theoremstyle{definition}
\newtheorem{exm}[thm]{Example}
\newtheorem{rmk}[thm]{Remark}

\DeclareMathOperator{\Tr}{Tr}
\DeclareMathOperator{\Id}{Id}
\let\Re\relax
\DeclareMathOperator{\Re}{Re}
\DeclareMathAlphabet{\mymathbb}{U}{bbold}{m}{n}

\newcommand{\stkout}[1]{\ifmmode\text{\sout{\ensuremath{#1}}}\else\sout{#1}\fi}
\newif\ifverbose
\verbosefalse
\newcommand{\ins}[1]{\ifverbose\textcolor{blue}{#1}\else#1\fi}

\title{Necessary Criteria for Markovian Divisibility of Linear Maps}

\author{Matthias C. Caro$^{1,2}$\footnote{ORCID: \href{https://orcid.org/0000-0001-9009-2372}{0000-0001-9009-2372}}~ and Benedikt R. Graswald$^1$\footnote{ORCID: \href{https://orcid.org/0000-0002-2777-2401}{0000-0002-2777-2401}} \\[1mm] 
\small $^1$ Technical University of Munich, Germany, Department of Mathematics \\
\small $^2$ Munich Center for Quantum Science and Technology (MCQST), Munich, Germany\\[-1mm]
\small {\tt \href{mailto:caro@ma.tum.de}{caro@ma.tum.de}, \href{mailto:graswabe@ma.tum.de}{graswabe@ma.tum.de} }}

\begin{document}

\maketitle

\begin{abstract}
Characterizing those quantum channels that correspond to Markovian time evolutions is an open problem in quantum information theory, even different notions of quantum Markovianity exist. One notion related to this problem is that of infinitesimal Markovian divisibility for quantum channels introduced in \cite{Wolf.2008}. Whereas there is a complete characterization for infinitesimal Markovian divisible qubit channels, no necessary or sufficient criteria are known for higher dimensions, except for necessity of non-negativity of the determinant.\\
We describe how to extend the notion of infinitesimal Markovian divsibility to general linear maps and compact and convex sets of generators. We give a general approach towards proving necessary criteria for (infinitesimal) Markovian divisibility that involve singular values of the linear map. With this approach, we prove two necessary criteria for infinitesimal divisibility of quantum channels that work in any finite dimension $d$: an upper bound on the determinant in terms of a $\Theta(d)$-power of the smallest singular value, and in terms of a product of $\Theta(d)$ smallest singular values. Our criteria allow us to analytically construct, in any given dimension, a set of channels that contains provably non infinitesimal Markovian divisible ones.\\
We also discuss the classical counterpart of this scenario, i.e., stochastic matrices with the generators given by transition rate matrices. Here, we show that no necessary criteria for infinitesimal Markovian divisibility of the form proved for quantum channels can hold in general. However, we describe subsets of all transition rate matrices for which our reasoning can be applied to obtain necessary conditions for Markovian divisibility.
\end{abstract}

\section{Introduction}

\cite{Gorini.1976} and \cite{Lindblad.1976} made an important step towards understanding the connection between master equations and the framework of quantum channels for describing quantum evolutions by characterizing the generators which give rise to semigroups of quantum channels via the corresponding (time-independent) master equation. The converse question, i.e., the problem of characterizing those quantum channels that can arise from the solution of a (possibly time-dependent) Lindblad master equation is, however, still awaiting an answer.\\

Endeavours towards a resolution of this problem have given rise to different notions of (Non-) Markovianity for quantum evolutions. One line of research is based on connecting Markovianity to certain divisibility properties of quantum evolutions, in particular to the possibility of dividing the evolution into infinitesimal pieces. While this gives an intuitively plausible notion of time-dependent quantum Markovianity and some structural properties can be established on its basis, it has so far not given rise to easily verifiable criteria for Markovianity (with a simple exception). Only for evolutions of qubit systems is this notion completely understood. We go beyond this characterization for the $2$-dimensional case and establish necessary criteria for a quantum channel - or a linear map in general - to be divisible into infinitesimal Markovian pieces. Our criteria take the form of an upper bound on the determinant in terms of a power of a product of smallest singular values.\\
Our proof strategy is not specific to quantum channels, but can be applied to obtain necessary criteria for (infinitesimal) Markovian divisibility of general linear maps w.r.t.~a closed and convex set of generators, if the generators satisfy certain spectral properties.

\subsection{Overview of our Results}
In this work, we study the following question: Given a linear map $T$ and a set of linear maps $\mathcal{G}$, acting on $\IC^d$, can $T$ be approximated arbitrarily well by linear maps of the form $\prod_i e^{G_i}$, where $G_i\in\mathcal{G}$? If that is the case, we say that $T$ is \emph{Markovian divisible w.r.t.~the set of generators} $\mathcal{G}$.\\

We aim towards establishing necessary criteria for Markovian divisibility of the form $$\lvert \det(T)\rvert \leq\left(\prod\limits_{i=1}^k s_i^\uparrow(T)\right)^p,$$ where $k=k(d)$ and $p=p(d)$ depend on the underlying dimension. Proving such criteria becomes tractable by combining multiplicativity of the determinant and sub-/super-multiplicativity of products of largest/smallest singular values with Trotterization.\\
In Section \ref{SctCriteriaGeneral}, we describe how to use these properties to reduce the problem of establishing necessary criteria of the above form to a spectral property of the generators. We can summarize our reduction as follows:
\begin{thm*}(Theorem \ref{ThmNecessaryCondMarkovianDivisibleGeneral} - Informal Version)\\
Let $\mathcal{G}\subseteq\mathcal{M}_d$ be a set of generators. Let $T$ be Markovian divisible w.r.t.~$\mathcal{G}$ and suppose that every $G\in\mathcal{G}$ satisfies $\Tr[G+G^\ast] - p \sum\limits_{i=1}^k \lambda_i^\uparrow (G+G^\ast)\leq 0$. Then $\lvert\det (T)\rvert\leq \left(\prod\limits_{i=1}^k s_i^\uparrow (T) \right)^p$.
\end{thm*}

We employ our proof strategy for the physically motivated scenario of \emph{infinitesimal Markovian divisibility}. Here, the objects of interest are linear maps $T$ that, for any $\varepsilon>0$, can be arbitrarily well approximated by linear maps of the form $\prod_i e^{G_i}$, where $G_i\in\mathcal{G}$ are s.t.~$\norm{e^{G_i}-\mathds{1}_d}\leq \varepsilon$.\\

We first study the case in which $\mathcal{G}$ is the set of Lindblad generators, seen as linear maps on $d\times d$-matrices. I.e., we consider those generators that give rise to semigroups of quantum channels. With this choice, the notion of infinitesimal Markovian divisibility of a linear map $T$ on $d\times d$-matrices becomes that of infinitesimal Markovian divisibility of quantum channels introduced in \cite{Wolf.2008}.\\
We prove necessary criteria for infinitesimal Markovian divisibility of quantum channels in any finite dimension. Namely, for an infinitesimal Markovian divisible quantum channel $T$ on $d\times d$-matrices we show in Corollaries \ref{CrlQuDetVSPowerSmallestSV} and \ref{CrlQuDetVSProductSmallestSV} that
\begin{align*}
    &|\det(T)|
    \leq
    \left(s_1^\uparrow (T)\right)^\frac{d}{2},
    \textrm{ and }~ 
    |\det (T)|
    \leq
    \prod\limits_{i=1}^{\lfloor 2d-2\sqrt{2d}+1\rfloor} s_i^{\uparrow}(T).
\end{align*}
Moreover, we give explicit examples (Examples \ref{ExmOptimalityLinearOrder} and \ref{ExmOptimalityLinearOrder2}) of infinitesimal divisible channels from which we can conclude that the $d$-dependence of the exponent (in the first bound) and of the number of singular value factors (in the second bound) is close to optimal, respectively.\\
We also describe  how to interpolate between these bounds in Corollary \ref{PropInterpolation} and obtain that for an infinitesimal divisible quantum channel $T$ acting on $d\times d$-matrices,
\begin{align*}
    |\det (T)|\leq \left( \prod\limits_{i=1}^{k} s_i^{\uparrow}(T)\right)^{\frac{2d}{k+2\sqrt{k}+1}},\quad\textrm{for }1\leq k\leq d^2.
\end{align*}
These criteria allow us to give new examples of provably non infinitesimal divisible channels in dimensions strictly bigger than $2$, which were not recognizable as such previously (Example \ref{ExmNewInfinitesimalDivisibleChannels}).\\

As a second application of our proof strategy, we take $\mathcal{G}$ to be the set of transition rate matrices of dimension $d$, and thereby study the question of (infinitesimal) Markovian divisibility of stochastic matrices. We first show via an explicit example (Example \ref{ExmCounterexampleStochastic}) that no necessary criterion of the above form can hold in this scenario when we allow all transition rate matrices as generators. Combined with our results for infinitesimal Markovian divisible quantum channels, this implies that stochastic matrices cannot be embedded into quantum channels while preserving both the singular values and the property of infinitesimal Markovian divisibility at the same time.\\
If, however, we restrict our set of generators to transition rate matrices whose diagonal elements differ by at most a constant factor, our proof strategy can be applied and yields an upper bound on the determinant in terms of a power of the smallest singular value (Corollary \ref{CrlStochDetVSPowerSmallestSV}).

\subsection{Related Work}
The quantum Markovianity problem, the question of deciding whether a given quantum channel is a member of a quantum dynamical semigroup, was considered from a complexity-theoretic perspective in \cite{Cubitt.2012}. Therein, it was shown to be NP-hard and that the same is true of the classical counterpart of this problem, with stochastic matrices instead of quantum channels and transition rate matrices instead of Lindblad generators. The computational complexity of a related divisibility problem for stochastic matrices, namely that of finite divisibility, was studied in \cite{Bausch.2016}. Also this divisibility problem turns out to be NP-hard, even NP-complete.\\
When fixing the system dimension, however, deciding whether a quantum channel is an exponential of a Lindblad generator, in which case it can be called time-independent Markovian because it solves a time-independent Lindblad master equation, becomes feasible. Corresponding necessary and sufficient criteria and an efficient (in the desired precision) algorithmic procedure for this case with a fixed dimension were given in \cite{Wolf.2008b, Cubitt.2012}. These results pertain to time-independent (quantum) Markovianity and cannot directly be applied to the time-dependent case.\\
Our focus is on infinitesimal Markovian divisibile of quantum channels. These were introduced and studied in detail for qubit channels by \cite{Wolf.2008}. Therein, it is also observed that every infinitely divisible quantum channel, i.e., every channel that can be written as an $n^{\textrm{th}}$ power of a quantum channel for every $n\in\mathbb{N}$, is infinitesimal divisible. The notion of infinitesimal Markovian divisibility can be seen as corresponding to time-dependent Markovianity, i.e., to solutions of time-dependent Lindblad master equations. Thereby, it offers a route of studying a time-dependent version of the Markovianity problem.\\

The plethora of different notions of Markovianity for quantum evolutions and relations between them are discussed in several review papers, among them \cite{Rivas.2014, Breuer.2016, Li.2018, Li.2019}. On the one hand, one considers notions of quantum Markovianity based on divisibility of the evolution, either for quantum channels or for quantum dynamical maps with corresponding propagators. This line of research was initiated by \cite{Wolf.2008}, references \cite{Davalos.2019, Chruscinski.2021} constitute recent additions to it. Related to this approach, reference \cite{Rivas.2010} proposed a measure of non-Markovianity on the basis of infinitesimal deviations from complete positivity. On the other hand, there are notions and measures of non-Markovianity based on (quantum) information backflow, often formalized in terms of distinguishability measures that are known to be non-increasing under completely positive and trace preserving maps. This idea was introduced in \cite{Breuer.2009}, reference \cite{Li.2018} recently proposed a variant of it. In Figure \ref{fig:NotionsMarkovianity}, we present only a selected few of these notions and of the connections between them.
\begin{figure}
    \centering
    \includegraphics[width = \textwidth]{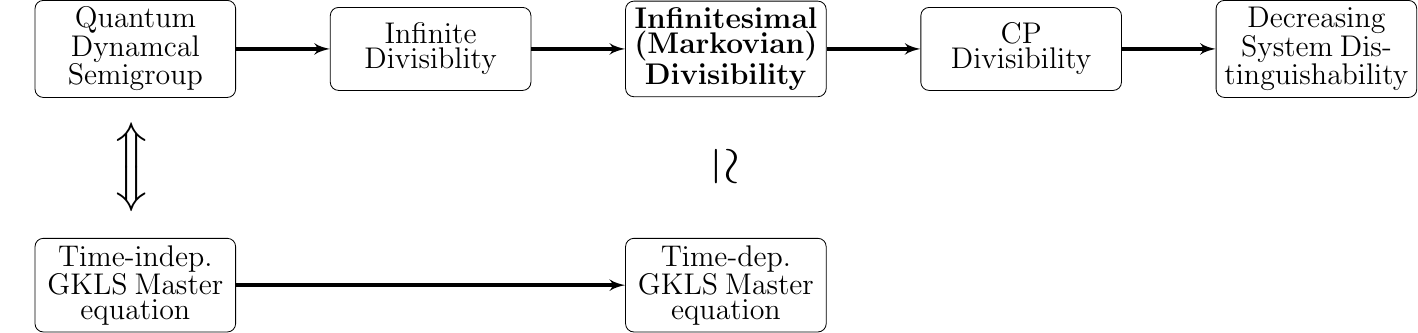}
    \caption{A depiction of the relations between different notions of divisibility and Markovianity of quantum channels and quantum dynamical maps. A simple arrow indicates that a channel or dynamical map satisfying the condition at the tail also satisfies that at the head. $\Updownarrow$ indicates the equivalence of two notions. And $\simeq$ is used to indicate a correspondence that, to the best of our knowledge, has been rigorously proven only for the qubit case.}
    \label{fig:NotionsMarkovianity}
\end{figure}


\subsection{Structure of the Paper}
Section \ref{SctPreliminaries} introduces basic notions from quantum information that provide our overall framework. In Section \ref{SctMarkovianDivisibility} we introduce the core definition of infinitesimal Markovian divisibility in a general setting and discuss prior work in the quantum scenario. Section \ref{SctCriteriaMarkovianDivisibility} contains our main results: We describe the general proof approach in Subsection \ref{SctCriteriaGeneral} and apply it to derive necessary criteria for infinitesimal Markovian divisibility of quantum channels in Subsection \ref{SctCriteriaQuantum}. The same type of criteria do not in general hold for infinitesimal divisibility of stochastic matrices, only for suitable subsets, as we argue in Subsection \ref{SctCriteriaStochastic}. We conclude with some open questions and the references. 

\section{Preliminaries}\label{SctPreliminaries}
We introduce some of the basic notions of quantum information, with a focus on quantum channels and the corresponding semigroups. The interested reader is referred to \cite{Nielsen.2010} for more details.\\

Throughout the paper, we denote the set of $d\times d$ complex matrices as $\mathcal{M}_d$, for a dimension $d\in\IN$. The identity matrix in $\mathcal{M}_d$ is written as $\mathds{1}_d$, whereas $\textrm{id}=\textrm{id}_{\mathcal{M}_d}$ denotes the identity map on $\mathcal{M}_d$. For $A\in\mathcal{M}_d$ we use $\lambda_i = \lambda_i(A)$ to denote its eigenvalues. If $A\in\mathcal{M}_d$ is Hermitian, we use $\lambda_i^\downarrow$ ($\lambda_i^\uparrow$) to denote the eigenvalues in decreasing (increasing) order. Similarly, we use the notation $s_i^\downarrow$ and $s_i^\uparrow$ for singular values. Finally, $\Tr[A]$ will denote the trace of $A$.

\subsection{Quantum States and Channels}

A $d$-level quantum system (for $d\in\IN$) is described by a $d\times d$ \emph{density matrix}, i.e., an element of $$\mathcal{S}\left( \IC^d\right):= \{ \rho\in\mathcal{M}_d ~|~ \rho\geq 0,\ \Tr[\rho]=1\},$$ where $\rho\geq 0$ means that the matrix $\rho$ is positive semidefinite.\\

Physically admissible transformations of quantum systems are described by \emph{quantum channels} (in the Schrödinger picture), i.e., by elements of $$\mathcal{T}\left( \IC^d, \IC^{d'}\right):= \{T:\mathcal{M}_d\to \mathcal{M}_{d'}~|~ T\textrm{ is linear, completely positive, and trace-preserving}\}.$$ Here, we call $T$ \emph{completely positive} iff $T\otimes \textrm{id}_{\mathcal{M}_n}$ is positivity-preserving for every $n\in\IN$. This definition guarantees that a quantum channel maps states to states and that this is still the case when embedding the quantum system of interest into a larger system with trivial evolution on the environmental subsystem.\\
We will also use the shorthand $\mathcal{T}_d := \mathcal{T}\left( \IC^d, \IC^{d}\right)$ for channels with equal input and output dimension. 

\subsection{Quantum Dynamical Semigroups}
It is a foundational postulate in quantum theory that the dynamics of a closed quantum system can be described in terms of a Schrödinger equation, which gives rise to a $1$-parameter group of unitaries. For open quantum systems, we will work with $1$-parameter semigroups.

\begin{dff}\emph{(Continuous dynamical semigroups)}\label{DffContDynSemigroup}\\
A family of linear maps $T_t:\mathcal{M}_d\to\mathcal{M}_d$ with time parameter $t\in\IR_+$ is called a \emph{dynamical semigroup} if $\forall t,s\in\IR_+ =[0,\infty) :T_tT_s=T_{t+s}$ and $T_0=Id$. If in addition the map $t\mapsto T_t$ is continuous (we are working on finite dimensional spaces, so there is no need to specify the type of continuity here), then the family is called a \emph{continuous dynamical semigroup}.
\end{dff}

It is well known that such continuous dynamical semigroups can be represented via a generator. I.e., if $\{ T_t\}_{t\geq 0}$ is a continuous dynamical semigroup, then there exists a linear map $L:\mathcal{M}_d\to\mathcal{M}_d$ s.t.~$T_t = e^{t L}$ for all $t\geq 0$.\\

When requiring such a semigroup to consist of physically admissible evolutions of a quantum system, i.e., of quantum channels, the question arises of what the corresponding generators are. This was answered in the following celebrated

\begin{thm}\emph{(Generators of quantum dynamical semigroups -- GKLS, \cite{Gorini.1976,Lindblad.1976})}\label{ThmGKLS}\\
A linear map $L:\mathcal{M}_d\to\mathcal{M}_d$ is the generator of a continuous dynamical semigroup of quantum channels if and only if it can be written as
\begin{align}
L(\rho) &= i[\rho,H] + \sum_j \L_j \rho \L_j^\dagger - \frac{1}{2}\lbrace \L_j^\dagger \L_j, \rho\rbrace, \label{eq::LindbladGenerator}
\end{align}
where $H=H^\dagger\in\mathcal{M}_d$ is self-adjoint and $\lbrace \L_j\rbrace_j$ is a set of matrices in $\mathcal{M}_d$. Here, $\{\cdot ,\cdot\}$ denotes the anti-commutator.
\end{thm}

For such generators, often called \emph{GKLS} or \emph{Lindblad generators}, we refer to the term $i[\cdot,H]$ as Hamiltonian part and to $\sum_j \L_j \cdot \L_j^\dagger - \frac{1}{2}\lbrace \L_j^\dagger \L_j, \cdot\rbrace$ as dissipative part with Lindbladians $\{ \L_j\}_j$.\\

We will call a quantum channel \emph{Markovian} if is an element of a quantum dynamical semigroup.

\section{Markovian Divisibility}\label{SctMarkovianDivisibility}

The main motivation for our work is the following problem: Given a quantum channel, decide whether it comes from a (possibly time-dependent) Lindblad master equation. We take two different perspectives on this task to motivate our definitions. \\
The first perspective is that of differential equations. Namely, we want to understand which quantum channels can arise as a solution of a time-dependent master equation of the form $\frac{\textrm{d}}{\textrm{dt}} T_t = L(t) T_t$, where $L(t)$ is a time-dependent Lindblad generator. More generally, we want to study the possible solutions of a linear ordinary differential equation $\frac{\textrm{d}}{\textrm{dt}} T_t = G(t) T_t$, where $t\mapsto G(t)\in\mathcal{G}$, with $\mathcal{G}\subset\mathcal{M}_d$ a fixed set of generators.\\
Our second perspective on the problem comes from the semigroup structure of the solutions to time-independent master equations. Namely, each such equation corresponds to a quantum dynamical semigroup. If we now also want to take into account a possible time-dependence of the generator while still preserving the semigroup structure, we can consider the semigroup generated by all elements of quantum dynamical semigroups. On an intuitive level, the question about solutions of master equations which we asked above now becomes the question of whether a given quantum channel is an element of this semigroup. I.e., we are dealing with the membership problem for this semigroup. And again, we can generalize the question by going from Lindblad generators to general generators.

\subsection{Markovian Divisibility w.r.t.~general Sets of Generators}
The two perspectives given above lead us to two slightly different definitions. In the first, we focus on the semigroup structure.

\begin{dff}\emph{(Markovian Divisibility)}\label{DffMarkovianDivisibleGeneral}\\
Let $\mathcal{G}\subset\mathcal{M}_d$ be a set of matrices, whose elements we call \emph{generators}. We define the set $$\mathcal{D}_\mathcal{G}:=\{ T\in\mathcal{M}_d ~|~\exists n\in\IN, \textrm{ generators }\{G_i\}_{1\leq i\leq n}\subset\mathcal{G} \textrm{ s.t.~} \prod\limits_{i=1}^n e^{G_i}=T \}.$$
We call the closure $\overline{\mathcal{D}}_\mathcal{G}$ the set of linear maps that are \emph{Markovian divisible w.r.t.~}$\mathcal{G}$. 
\end{dff}

When translating the mathematical motivation of semigroups to a more physical motivation, Definition \ref{DffMarkovianDivisibleGeneral} can be seen as an approach to the question of which linear maps can be arbitrarily well approximated using alternating exponentials of a fixed set of (\emph{control}) generators.\\

Now, we give a definition based more on the perspective of differential equations determining the overall evolution on infinitesimal time intervals, while keeping the semigroup structure in mind.

\begin{dff}\emph{(Infinitesimal Markovian Divisibility)}\label{DffInfinitesimalDivisibilityGeneral}\\
Let $\mathcal{G}\subset\mathcal{M}_d$ be a compact and convex set of matrices containing $0\in\mathcal{M}_d$. We will again refer to elements of $\mathcal{G}$ as \emph{generators}. We define the set 
\begin{align*}
\mathcal{I}_\mathcal{G} 
:= \{T\in\mathcal{M}_d ~|~ &\forall\varepsilon >0 ~\exists n\in\IN, \textrm{ generators }\{G_j\}_{1\leq j\leq n}\subset\mathcal{G} \\ &\textrm{s.t.~}(i)\norm{e^{G_j} - \mathds{1}_d}\leq\varepsilon~\forall j\textrm{ and }(ii)\prod\limits_{j=1}^n e^{G_j}=T \}.
\end{align*}
We call the closure $\overline{\mathcal{I}}_\mathcal{G}$ the set of linear maps that are \emph{infinitesimal Markovian divisible w.r.t.~}$\mathcal{G}$. 
\end{dff}

\begin{rmk}\label{RmkCompactnessWLOG}
In the definition, we require $\mathcal{G}$ to be compact. This can be assumed w.l.o.g.~First, closedness can be assumed w.l.o.g.~since for non-closed $\mathcal{G}_0$ we have $\overline{\mathcal{I}}_{\mathcal{G}_0} =\overline{\mathcal{I}}_{\overline{\mathcal{G}_0}}$. Second, boundedness can also be assumed w.l.o.g.~Namely, suppose $\tilde{\mathcal{G}}\subset\mathcal{M}_d$ is an unbounded closed and convex set with $0\in\tilde{\mathcal{G}}$ and $T\in\mathcal{I}_{\tilde{\mathcal{G}}}$. Then, by definition, $\forall\varepsilon >0$ $\exists n\in\IN$ and $\{G_j\}_{1\leq j\leq n}\subset\tilde{\mathcal{G}}$ s.t.~$\norm{e^{G_j} - \mathds{1}_d}\leq\varepsilon$ and $\prod_{j=1}^n e^{G_j}=T$. By convexity, also $\frac{1}{N}G_j\in\tilde{\mathcal{G}}~\forall1\leq j\leq n$ for every $N>1$. By continuity of the matrix exponential, there exists $N_0\in\IN$ s.t.~$\norm{e^{\frac{1}{N} G_j} - \mathds{1}_d}\leq\varepsilon$ for all $N\geq N_0$. And clearly we can write $T=\prod_{j=1}^n e^{G_j} = \prod_{j=1}^n \left(e^{\frac{1}{N}G_j}\right)^N$. Thus, as $\norm{\frac{1}{N}G_j}\to 0$ as $N\to\infty$, we conclude that for every $B>0$ we have $T\in \mathcal{I}_{\tilde{\mathcal{G}}_{\leq B}}$, where $\tilde{\mathcal{G}}_{\leq B}:=\{G\in \tilde{\mathcal{G}}~|~\norm{G}\leq B\}$. Hence, we can impose an arbitrary (non-zero) norm bound on our generators without changing the set of infinitesimal Markovian divisible channels.\\
Therefore, we are justified in using Definition \ref{DffInfinitesimalDivisibilityGeneral} also for non-compact $\mathcal{G}$ (in particular, Lindblad generators and transition rate matrices).
\end{rmk}





\begin{rmk}\label{RmkInfinitesimalMarkovianversusMarkovianDivisible}
By continuity of the matrix exponential, it is easy to see that, if $G\in\mathcal{G}$ implies $\frac{1}{n}G\in\mathcal{G}$ for all $n\in\IN$, then $\mathcal{D}_\mathcal{G}=\mathcal{I}_\mathcal{G}$. This is in particular the case if $\mathcal{G}$ satisfies the assumptions of Definition \ref{DffInfinitesimalDivisibilityGeneral}.\\
If, however, $\mathcal{G}$ does not have this property, then $(i)$ in the definition of $\mathcal{I}_\mathcal{G}$ will in general lead to $\mathcal{I}_\mathcal{G}\neq\mathcal{D}_\mathcal{G}$. (E.g., $\mathcal{I}_\mathcal{G}$ could be empty even if $\mathcal{D}_\mathcal{G}$ is not).
\end{rmk}

When specifying $\mathcal{G}$ to be the set of Lindblad generators, and thus the linear maps of interest to be quantum channels, Definitions \ref{DffMarkovianDivisibleGeneral} and \ref{DffInfinitesimalDivisibilityGeneral} become connected to quantum channels arising from master equations. Studying such channels via a notion of Markovian divisibility into infinitesimal pieces was first proposed in \cite{Wolf.2008}. Next, we discuss some results of that work.

\subsection{Infinitesimal Markovian Divisibility of Quantum Channels}\label{SctMarkovianDivisibilityQuantum}
For ease of notation, we will denote by $\mathcal{I}_d$ the set $\mathcal{I}_\mathcal{G}$ for the specific choice of $\mathcal{G}$ being the set of Lindblad generators acting on $d\times d$-matrices. (This is an unbounded set of generators, but we can nevertheless use it in Definition \ref{DffInfinitesimalDivisibilityGeneral} by Remark \ref{RmkCompactnessWLOG}.) Then, the set $\overline{\mathcal{I}}_d$ is the set of \emph{infinitesimal Markovian divisible quantum channels} as defined in \cite{Wolf.2008}.

When referring to these channels, we will sometimes drop the ``Markovian'' for convenience. This can also be justified in a rigorous sense, see Theorem $16$ in \cite{Wolf.2008}.\\



While some insight into the structure of infinitesimal Markovian divisible quantum channels has been obtained in \cite{Wolf.2008}, so far there are no simple-to-check criteria for infinitesimal divisibility for a general dimension $d$. Such criteria are the main focus of this work.\\

A straightforward necessary criterion for infinitesimal divisibility is already observed in \cite{Wolf.2008}, namely we have as a direct consequence of multiplicativity and continuity of the determinant:

\begin{prp}\label{PrpNonnegDet}
An infinitesimal divisible quantum channel $T$ satisfies $\det (T) \geq 0$.
\end{prp}

This is, to our knowledge, the only necessary criterion for infinitesimal divisibility known so far that holds in any finite dimension.\\

For the special case of qubit channels, the set of infinitesimal divisible channels can be explicitly characterized making use of the Lorentz normal form. (The latter is discussed, e.g., in \cite{Verstraete.20030122}).


\begin{thm}\emph{(Infinitesimal divisible qubit channels \cite{Wolf.2008} - Informal)}\label{ThmQubitInfinitesimalDivisible}\\
Let $T:\mathcal{M}_2\to\mathcal{M}_2$ be a generic qubit channel with Lorentz normal form 
$\begin{pmatrix}
1 & 0\\
0 & \Delta
\end{pmatrix}$.\\
$T$ is infinitesimal Markovian divisible if and only if $0\leq \det(\Delta)\leq \textrm{s}_{\textrm{min}}^2$, where $\textrm{s}_{\textrm{min}}$ is the smallest singular value of $\Delta$.
\end{thm}

This characterization serves as one motivation for our results in higher dimensions, which we derive in Subsection \ref{SctCriteriaQuantum}.

\section{Necessary Criteria for Markovian Divisibility}\label{SctCriteriaMarkovianDivisibility}
We now develop necessary criteria for a linear map to be (infinitesimal) Markovian divisible. More precisely, our discussion aims towards establishing inequalities of the form 
\begin{align}
\lvert \det (T)\rvert\leq \left(\prod\limits_{i=1}^k s_i^\uparrow (T) \right)^p.\label{eq:GoalInequality}
\end{align} We first present some results for the case of general linear maps \& generators and later combine these observations with a more detailed analysis for quantum channels \& Lindblad generators and stochastic matrices \& transition rate matrices, respectively.

\subsection{General Sets of Generators}\label{SctCriteriaGeneral}
We first observe that if each of two matrices satisfies the desired inequality \eqref{eq:GoalInequality}, then so does the product of the matrices.
\begin{lmm}\label{LmmReductionSingleFactors}
Let $T_1,T_2\in\mathcal{M}_d$. Suppose that $1\leq k\leq d$ and $p>0$ are s.t.~$$\lvert \det(T_j)\rvert\leq \left(\prod\limits_{i=1}^k s_i^\uparrow (T_j) \right)^p $$ holds for $j=1,2$. Then also $$ \lvert\det(T_1T_2)\rvert\leq \left(\prod\limits_{i=1}^k s_i^\uparrow (T_1T_2) \right)^p.$$
\end{lmm}
\begin{proof}
A well-known majorisation inequality for singular values states that
\begin{align}
    \prod\limits_{i=1}^k s_i^{\downarrow}(AB) \leq \prod\limits_{i=1}^k s_i^{\downarrow}(A)s_i^{\downarrow}(B) \label{eq::SVSubmult}
\end{align}
for any $1\leq k\leq n$ for $n\times n$-matrices $A,B$  (see, e.g., \cite{Horn.1991}, Theorem $3.3.4$). With this we obtain
\begin{align*}
\lvert\det (T_1T_2) \rvert
&= \lvert\det(T_1)\rvert \lvert\det(T_2)\rvert\\
&\leq\left(\prod\limits_{i=1}^k s_i^\uparrow (T_1) \right)^p \left(\prod\limits_{i=1}^k s_i^\uparrow (T_2) \right)^p\\
&=\left(\frac{\lvert\det(T_1)\rvert\lvert\det(T_2)\rvert}{\prod\limits_{i=1}^{d-k}s_i^\downarrow (T_1)s_i^\downarrow(T_2)} \right)^p\\
&\leq\left(\frac{\lvert\det(T_1T_2)\rvert}{\prod\limits_{i=1}^{d-k}s_i^\downarrow (T_1T_2)} \right)^p\\
&=\left(\prod\limits_{i=1}^{k}s_i^\uparrow (T_1T_2) \right)^p,
\end{align*}
as claimed. Here, the first inequality is by assumption, the following step uses $\lvert\det(T_i)\rvert=\prod\limits_{j=1}^{d} s_j^\downarrow (T_i)$, the second inequality is due to Equation \eqref{eq::SVSubmult}, and the last step uses $\lvert\det(T_1T_2)\rvert=\prod\limits_{j=1}^{d} s_j^\downarrow (T_1T_2)$.
\end{proof}

This means that, when trying to establish an inequality of the form \eqref{eq:GoalInequality}, if $T$ is a finite product, it suffices to consider the single factors separately.\\

Now we show that, once we have our desired inequality \eqref{eq:GoalInequality} for non-negative multiples of two separate generators, the exponential of the sum of these two generators also satisfies the inequality. This observation will be particularly useful in our analysis of Lindblad generators.

\begin{lmm}\label{LmmSumofGenerators}
Let $G_1,G_2\in\mathcal{M}_d$. Suppose that $1\leq k\leq d$ and $p>0$ are s.t.~$$\lvert \det(e^{\frac{G_j}{n}})\rvert\leq \left(\prod\limits_{i=1}^k s_i^\uparrow (e^{\frac{G_j}{n}}) \right)^p $$ holds for all $n\in\IN$ and $j=1,2$.
Then also $$ \lvert\det(e^{G_1 + G_2})\rvert\leq \left(\prod\limits_{i=1}^k s_i^\uparrow (e^{G_1 + G_2}) \right)^p.$$
\end{lmm}
\begin{proof}
By the Lie-Trotter formula, $e^{A+B}=\lim_{n\to\infty} (e^{\frac{A}{n}}e^{\frac{B}{n}})^n$. As both the determinant and the singular values depend continuously on the matrix, we can combine this with (an iterative application of) Lemma \ref{LmmReductionSingleFactors} to see that it suffices to have $\lvert\det (e^{\frac{G_i}{n}})\rvert\leq \left(\prod\limits_{i=1}^k s_i^\uparrow (e^{\frac{G_i}{n}}) \right)^p$ for arbitrary $n\in\IN$. We can summarize this reasoning as follows:
\begin{align*}
    \lvert\det(e^{G_1 + G_2})\rvert
    &= \lim\limits_{n\to\infty}\lvert\det ((e^{\frac{G_1}{n}}e^{\frac{G_2}{n}})^n)\rvert\\
    &\leq\lim\limits_{n\to\infty}\left(\prod\limits_{i=1}^k s_i^\uparrow ((e^{\frac{G_1}{n}}e^{\frac{G_2}{n}})^n) \right)^p\\
    &= \left(\prod\limits_{i=1}^k s_i^\uparrow (e^{G_1 + G_2}) \right)^p,
\end{align*}
where the inequality follows by combining the assumption with Lemma \ref{LmmReductionSingleFactors}.
\end{proof}

\begin{rmk}\label{RmkPrefactorsSmallerThan1}
If $G_j$ in Lemma \ref{LmmSumofGenerators} are normal matrices, then it is easy to see that the assumed inequality for $n=1$ already implies the corresponding inequality for any $n\in\IN$. In general, however, this implication is not true. This can be seen, e.g., by considering $L$ and $\frac{1}{2}L$, with $L$ as given in Example \ref{ExmOptimalityLinearOrder}. Therefore, we make the assumption for all $n\in\IN$. This is also why we formulate Definition \ref{DffInfinitesimalDivisibilityGeneral} for convex sets of generators that contain the $0$-matrix.
\end{rmk}


Next, we discuss how to reduce an inequality of the form \eqref{eq:GoalInequality} for a single matrix exponential to an inequality of eigenvalues of the exponent.

\begin{lmm}\label{LmmSufficientEigAssumptionGeneral}
Suppose that $G\in\mathcal{M}_{d}$ satisfies $\Tr[G+G^\ast] - p \sum\limits_{i=1}^k \lambda_i^\uparrow (G+G^\ast)\leq 0$, then \[
\lvert \det (e^G)\rvert \leq \left(\prod\limits_{i=1}^k s_i^\uparrow (e^G)\right)^p.
\]
\end{lmm}
\begin{proof}
We observe that 
\begin{align*}
    \prod\limits_{i=1}^k s_i^\uparrow (e^G)
    &= \frac{\lvert\det (e^G)\rvert}{\prod\limits_{i=1}^{d-k} s_i^\downarrow (e^G)}
    \geq \frac{\lvert\det (e^G)\rvert}{\prod\limits_{i=1}^{d - k} s_i^\downarrow (e^{\tfrac{1}{2}(G+G^\ast)})}
    = \frac{\det (e^{\tfrac{1}{2} (G+G^\ast)})}{\prod\limits_{i=1}^{d-k} e^{\tfrac{1}{2}\lambda_i^\downarrow (G+G^\ast)}}
    = \prod\limits_{i=1}^k e^{\tfrac{1}{2}\lambda_i^\uparrow (G+G^\ast)},
\end{align*}
where we used $\prod\limits_{i=1}^{d-k} s_i^{\downarrow}(e^G) \leq \prod\limits_{i=1}^{d-k} s_i^{\downarrow}(e^{\Re(G)})$ (see, e.g., \cite{Bhatia.1997}, p.$259$), as well as $\lvert\det (e^G)\rvert=\det (e^{\tfrac{1}{2} (G+G^\ast)})$, which can be seen via Lie-Trotter. With this we now obtain 
\begin{align*}
    |\det (e^G)|^2
    &= e^{\Tr[G+G^\ast]}
    \leq \left(e^{\sum\limits_{i=1}^k \lambda_i^\uparrow (G+G^\ast)}\right)^p
    = \left(\prod\limits_{i=1}^{k} e^{\tfrac{1}{2}\lambda_i^\uparrow (G+G^\ast)} \right)^{2p}
    \leq \left( \prod\limits_{i=1}^k s_i^\uparrow (e^G)\right)^{2p},
\end{align*}
where the first inequality is exactly our assumption. Now we take the square root and obtain the claimed inequality.
\end{proof}

We summarize the results of the foregoing discussion for Markovian divisibility in the following 
\begin{thm}\label{ThmNecessaryCondMarkovianDivisibleGeneral}
Let $\mathcal{G}\subseteq\mathcal{M}_d$ be a set of generators. Let $T\in\overline{\mathcal{D}}_\mathcal{G}$ and suppose that every $G\in\mathcal{G}$ satisfies $\Tr[G+G^\ast] - p \sum\limits_{i=1}^k \lambda_i^\uparrow (G+G^\ast)\leq 0$. Then $\lvert\det (T)\rvert\leq \left(\prod\limits_{i=1}^k s_i^\uparrow (T) \right)^p$.
\end{thm}
\begin{proof}
By continuity of the determinant and the singular values, we can restrict our attention to $T\in\mathcal{D}_G$. In that case, there exist $n\in\IN$ and generators $\{G_i\}_{1\leq i\leq n}\subset\mathcal{G}$ s.t.~$\prod\limits_{i=1}^n e^{G_i}=T$. By Lemma \ref{LmmReductionSingleFactors}, it suffices to have the desired inequality for each factor $e^{G_i}$. These now satisfy the inequality by Lemma \ref{LmmSufficientEigAssumptionGeneral}.
\end{proof}

We obtain an analogous result for infinitesimal Markovian divisibility:
\begin{crl}\label{CrlNecessaryCondInfinitesimalMarkovianDivisibleGeneral}
Let $\mathcal{G}\subset\mathcal{M}_d$ be a compact and convex set of matrices containing $0\in\mathcal{M}_d$. Let $\tilde{\mathcal{G}}:=\{\lambda G~|~\lambda\in[0,1],~G\textrm{ an extreme point of }\mathcal{G}\}\subset\mathcal{G}$. Assume that every $\tilde{G}\in\tilde{\mathcal{G}}$ satisfies $\Tr[\tilde{G}+\tilde{G}^\ast] - p \sum\limits_{i=1}^k \lambda_i^\uparrow (\tilde{G}+\tilde{G}^\ast)\leq 0$. Let $T\in\overline{\mathcal{I}}_\mathcal{G}$. Then $0\leq \det (T)\leq \left(\prod\limits_{i=1}^k s_i^\uparrow (T) \right)^p$.
\end{crl}
\begin{proof}
$\det(T)\geq 0$ follows in the same way as in Proposition \ref{PrpNonnegDet}. By continuity it suffices to prove the desired upper bound for $T\in\mathcal{I}_G$. By the definition of the set $\mathcal{I}_G$ and Lemma \ref{LmmReductionSingleFactors}, it then suffices to consider single factors of the form $e^G$, $G\in\mathcal{G}$. By definition of $\tilde{\mathcal{G}}$, $\tilde{G}\in\tilde{\mathcal{G}}$ in particular implies $\frac{1}{n}\tilde{G}\in\tilde{\mathcal{G}}$ for all $n\in\IN$. Also, every element of $\mathcal{G}$ can be expressed as a finite sum of elements of $\tilde{\mathcal{G}}$ (by Krein-Milman). Therefore, we can apply Lemma \ref{LmmSumofGenerators} to conclude that it suffices to consider single factors of the form $e^{\tilde{G}}$, $\tilde{G}\in\tilde{\mathcal{G}}$. Now we apply Lemma \ref{LmmSufficientEigAssumptionGeneral} to finish the proof.
\end{proof}

The assumption in Corollary \ref{CrlNecessaryCondInfinitesimalMarkovianDivisibleGeneral} is about (truncated) rays through extreme points of the convex set of interest. In light of Remark \ref{RmkPrefactorsSmallerThan1}, we expect that this can in general not be further simplified to an assumption only about the extreme points themselves (without multiples).

\subsection{Quantum Channels}\label{SctCriteriaQuantum}
We now want to apply the reasoning from the previous subsection to the more specific question of infinitesimal (Markovian) divisibility of quantum channels.\\
To avoid confusion about notation, in this subsection we will denote the eigenvalues of a matrix $\L$ as $\lambda_i=\lambda_i(\L)$, whereas the eigenvalues of a linear map $L$ on matrices are written as $\Lambda_K=\Lambda_K (L)$. For real eigenvalues of such linear superoperators, we use $\Lambda_K^\downarrow$ ($\Lambda_K^\uparrow$) to denote the eigenvalues in decreasing (increasing) order.\\

\subsubsection{Determinant versus power of the smallest singular value}
We first show that purely dissipative Lindblad generators with one Lindbladian satisfy an inequality as assumed in Lemma \ref{LmmSufficientEigAssumptionGeneral} with only one summand:
\begin{lmm}\label{LmmEigIneqQuantumPowerCriterion}
Let $L:\mathcal{M}_d\to\mathcal{M}_d$, $L(\rho)=\mathcal{L}\rho\mathcal{L}^\dagger - \frac{1}{2}\{\mathcal{L}^\dagger \mathcal{L},\rho\}$ be a purely dissipative Lindblad generator with one Lindbladian $\mathcal{L}\in\mathcal{M}_d$. Then 
\begin{align}
    \Tr[L+L^\ast]-\frac{d}{2}\Lambda_1^\uparrow (L+L^\ast)\leq 0.\label{eq::EVIneqLindbladGenV1}
\end{align}
\end{lmm}
\begin{proof}
We adopt the following convention for vectorization of matrices: If $A$ is an $n\times n$-matrix with column vectors $a_i$, then $vec(A)=(a_1^T,\ldots,a_n^T)^T$ is the column vector obtained by stacking the columns of $A$ on top of one another. When using $vec(ABC) = (C^T\otimes A) vec(B)$ to rewrite $L+L^\ast$ as a $d^2\times d^2$-matrix we obtain $$vec(L+L^\ast) = \overline{\mathcal{L}}\otimes\mathcal{L} + \overline{\mathcal{L}^\dagger}\otimes\mathcal{L}^\dagger - \mathds{1}_d\otimes\mathcal{L}^\dagger\mathcal{L} - \overline{\mathcal{L}^\dagger\mathcal{L}}\otimes \mathds{1}_d.$$ 
From this it is easy to see that $$\Tr[L+L^\ast]=\lvert\Tr[\L]\rvert^2 - 2d\norm{\L}_F^2.$$
We observe that the Lindbladians $\L$ and $\lambda \mathds{1}_d + \L$ give rise to the same superoperator $L+L^\ast$ for every $\lambda\in\IC$. So we can w.l.o.g.~assume that $\Tr[\L]=0$ and therefore $\Tr [L + L^\ast]=- 2d\norm{\L}_F^2$. Thus we obtain
\begin{align*}
    \Tr[L+L^\ast]-\frac{d}{2}\Lambda_1^\uparrow (L+L^\ast)
    &\leq - 2d\norm{\L}_F^2 + \frac{d}{2}\norm{L+L^\ast}_\infty\\
    &\leq - 2d\norm{\L}_F^2 + \frac{d}{2}\left(\norm{\overline{\mathcal{L}}\otimes\mathcal{L}}_\infty + \norm{\overline{\mathcal{L}^\dagger}\otimes\mathcal{L}^\dagger}_\infty + \norm{\mathds{1}_d\otimes\mathcal{L}^\dagger\mathcal{L}}_\infty + \norm{\overline{\mathcal{L}^\dagger\mathcal{L}}\otimes \mathds{1}_d}_\infty\right)\\
    &= - 2d\norm{\L}_F^2 + \frac{d}{2}\cdot 4\norm{\mathcal{L}}_\infty^2\\
    &\leq 0,
\end{align*}
which finishes the proof.
\end{proof}

\begin{rmk} \label{remark:improvement}
In our proof of Lemma \ref{LmmEigIneqQuantumPowerCriterion}, one step might strike the reader as particularly simplistic. Namely, we estimate $$\frac{d}{2}\norm{L+L^\ast}_\infty 
\leq \frac{d}{2}\left(\norm{\overline{\mathcal{L}}\otimes\mathcal{L}}_\infty + \norm{\overline{\mathcal{L}^\dagger}\otimes\mathcal{L}^\dagger}_\infty + \norm{\mathds{1}_d\otimes\mathcal{L}^\dagger\mathcal{L}}_\infty + \norm{\overline{\mathcal{L}^\dagger\mathcal{L}}\otimes \mathds{1}_d}_\infty\right) 
\leq \frac{d}{2}\cdot 4\norm{\L}_\infty^2.$$
With a more thorough analysis, we can slightly improve this upper bound and thereby increase the prefactor in the statement of Lemma \ref{LmmEigIneqQuantumPowerCriterion} from $\frac{d}{2}$ to $\approx 0.610733~d$. (We then get the same improvement in Corollary \ref{CrlQuDetVSPowerSmallestSV} below.) We derive this improvement in Appendix \ref{SctAppendixImprovement}.
\end{rmk}

We can now apply the reasoning from the previous subsection (for $k=1$ and $p=\frac{d}{2}$) to obtain
\begin{crl}\label{CrlQuDetVSPowerSmallestSV}
Let $T\in\overline{\mathcal{I}}_d$. Then $0\leq \det(T)\leq \left(s_1^\uparrow (T)\right)^\frac{d}{2}.$
\end{crl}
\begin{proof}
By combining the form of Lindblad generators from Theorem \ref{ThmGKLS} with Corollary \ref{CrlNecessaryCondInfinitesimalMarkovianDivisibleGeneral}, it suffices to consider Lindblad generators with a single summand, i.e., of the form
\begin{align*}
    L(\rho) = \begin{cases} i[\rho, H] \text{ with } H=H^\dagger \\ \mathcal{L} \rho\mathcal{L}^{\dagger}- \frac{1}{2} \lbrace \mathcal{L}^{\dagger} \mathcal{L},\rho\rbrace \end{cases}.
\end{align*}
$[\cdot,H]:\mathcal{M}_d\to\mathcal{M}_d$ is a self-adjoint map if $H=H^\dagger$ and therefore $e^{i[\cdot, H]}$ has $1$ as only singular value. The desired singular value inequality \eqref{eq:GoalInequality} is thus trivially satisfied for factors of this form. For factors of the form $e^L$ with $L(\rho)=\mathcal{L} \rho\mathcal{L}^{\dagger}- \frac{1}{2} \lbrace \mathcal{L}^{\dagger} \mathcal{L},\rho\rbrace$, the desired eigenvalue inequality is exactly shown in Lemma \ref{LmmEigIneqQuantumPowerCriterion}.
\end{proof}

This necessary criterion can be used to find channels that are not infinitesimal divisible and are given by convex combinations of a rank-deficient channel with the identity channel.

\begin{crl}\label{CrlNonInfinitesimalDivisibleChannelsNeighbourhood}
Let $T:\mathcal{M}_d\to\mathcal{M}_d$ be a quantum channel that has singular value $0$ of multiplicity $1\leq k < \frac{d}{2}$. Then every neighbourhood of $T$ contains a non infinitesimal divisible channel.
\end{crl}
\begin{proof}
Given such a quantum channel $T$ we can explicitly write down non infinitesimal divisible channels via convex combination with the identity, $T_\epsilon = (1-\epsilon) T + \epsilon \Id$.
By assumption, $T_\epsilon$ has exactly $k$ singular values which go to 0 as $\epsilon \to 0$. 
Thus either $\det(T_\epsilon)<0$ or we have
\begin{align*}
    \det(T_\epsilon) 
    = \prod_{j=1}^{d^2} s_j^{\uparrow}(T_\epsilon)
    \geq
    \big(  s_1^{\uparrow}(T_\epsilon)\big)^k 
    \prod_{j=k+1}^{d^2} s_j^{\uparrow}(T_\epsilon)
    >
     \big( s_1^{\uparrow}(T_\epsilon)\big)^{\nicefrac{d}{2}}, \quad \text{for $\epsilon$ small enough},
\end{align*}
where we just used that the $d^2-k$ largest singular values do not go to 0 for $\epsilon \to 0$.
Hence, for $\epsilon >0$ small enough, $T_\epsilon$ does not satisfy the criterion given in Corollary \ref{CrlQuDetVSPowerSmallestSV} and is therefore not infinitesimal divisible.
\end{proof}

\begin{exm}\label{ExmNewInfinitesimalDivisibleChannels}
We can use the above Corollary to find \ins{non-}infinitesimal divisible channels near the channel $T:\mathcal{M}_d\to\mathcal{M}_d$, $T(\rho)=\frac{\Tr[\rho]}{d}\mathds{1}_d$. $T$ is diagonal w.r.t.~the generalized Gell-Mann basis of $\mathcal{M}_d$ with the corresponding matrix given by $\hat{T}=\textrm{diag}[1,0,0,\ldots,0]$. The Choi matrix $\tau$ of $T$ has full rank and is thus in particular strictly positive definite (because complete positivity of $T$ translates to positive semidefiniteness of its Choi matrix $\tau$, see, e.g., \cite{Nielsen.2010}).\\
So we can pick $\varepsilon >0$ small enough s.t.~$\hat{T}_\varepsilon=\textrm{diag}[1,\varepsilon,\ldots,\varepsilon,0]$ is the matrix representation of a completely positive map in the generalized Gell-Mann basis. As such a matrix $\hat{T}_\varepsilon$ describes by its very form a trace-preserving map, it corresponds to a quantum channel $T_\varepsilon$ which now has eigenvalue $0$ of multiplicity $1$. So we can apply Corollary \ref{CrlNonInfinitesimalDivisibleChannelsNeighbourhood} to $T_\varepsilon$ and thus find channels arbitrarily close to $T$ that are not infinitesimal divisible.
\end{exm}

Naturally, the question arises whether the power $\frac{d}{2}$ in Corollary \ref{CrlQuDetVSPowerSmallestSV} is optimal. Our next example shows that the dependence on $d$ cannot be better than linear and that the factor of $\frac{1}{2}$ cannot be improved by much.

\begin{exm}\label{ExmOptimalityLinearOrder}
When considering the pathological case of a matrix of the form 
\begin{align*}
    \L = \begin{pmatrix}
    0 & 0 & 0 & \ldots & 1\\
    0 & 0 & 0 & \ldots & 0\\
    \vdots & & \ddots & & \vdots\\
    0 & 0 & 0 & \ldots & 0\\
    \end{pmatrix},
\end{align*}
we can easily compute that
\begin{align*}
    L+L^\ast
    = 
    \begin{pmatrix}
    0 & 0 & \ldots & 1\\
    0 & 0 & \ldots & 0\\
    \vdots & & \ddots  & \vdots\\
    1 & 0 & \ldots & 0
    \end{pmatrix}
    +
    \begin{pmatrix}
    D_1 & 0 & \ldots & 0\\
    0 & D_2 &  & 0\\
    \vdots & & \ddots & \vdots\\
    0 & 0 & \ldots & D_d
    \end{pmatrix}
\end{align*}
with $D_i = \textrm{diag}(0,\ldots,0,-1)\in\IR^{d\times d}$ for $1\leq i\leq d-1$ and $D_{d}=\textrm{diag}(-1,\ldots,-1,-2)\in\IR^{d\times d}$. So $L+L^\ast$ has eigenvalues $-1$ of multiplicity $2(d-1)$, $0$ of multiplicity $d^2-2d$, and $-1\pm\sqrt{2}$, each of multiplicity $1$. In particular, $\Tr[L+L^\ast] - p\Lambda_1^\uparrow (L+L^\ast) =-2d +(1+\sqrt{2})p\leq 0$ iff $p\leq \frac{2}{1+\sqrt{2}}d$.\\

This example also shows that in Theorem \ref{CrlQuDetVSPowerSmallestSV} nothing better than $\det(T)\leq \left( s_1^\uparrow (T)\right)^p$ with $p=\mathcal{O}(d)$ can be achieved. Namely, with the above choice of $\L$ we get
\begin{align*}
    L =     \begin{pmatrix}
    0 & 0 & \ldots & 1\\
    0 & 0 & \ldots & 0\\
    \vdots & & \ddots  & \vdots\\
    0 & 0 & \ldots & 0
    \end{pmatrix}
    + \frac{1}{2}
    \begin{pmatrix}
    D_1 & 0 & \ldots & 0\\
    0 & D_2 &  & 0\\
    \vdots & & \ddots & \vdots\\
    0 & 0 & \ldots & D_d
    \end{pmatrix}.
\end{align*}
This can now be exponentiated to obtain
\begin{align*}
    T := e^{L} = \begin{pmatrix}
    0 & 0 & \ldots & 1-e^{-1}\\
    0 & 0 & \ldots & 0\\
    \vdots & & \ddots  & \vdots\\
    0 & 0 & \ldots & 0
    \end{pmatrix}
    + 
    \begin{pmatrix}
    e^{\tfrac{1}{2}D_1} & 0 & \ldots & 0\\
    0 & e^{\tfrac{1}{2}D_2} &  & 0\\
    \vdots & & \ddots & \vdots\\
    0 & 0 & \ldots & e^{\tfrac{1}{2}D_d}
    \end{pmatrix},
\end{align*}
where $e^{\nicefrac{1}{2}D_i}=\textrm{diag}(1,\ldots,1,e^{-\nicefrac{1}{2}})$ for $1\leq i\leq d-1$ and $e^{\nicefrac{1}{2}D_d}=\textrm{diag}(e^{-\nicefrac{1}{2}},\ldots,e^{-\nicefrac{1}{2}},e^{-1})$.\\
We can now compute
\begin{align*}
    T^* T = \begin{pmatrix}
    0 & 0 & \ldots & 1-e^{-1}\\
    0 & 0 & \ldots & 0\\
    \vdots & & \ddots  & \vdots\\
    1-e^{-1} & 0 & \ldots & (1-e^{-1})^2
    \end{pmatrix}
    + 
    \begin{pmatrix}
    e^{D_1} & 0 & \ldots & 0\\
    0 & e^{D_2} &  & 0\\
    \vdots & & \ddots & \vdots\\
    0 & 0 & \ldots & e^{D_d}
    \end{pmatrix},
\end{align*}
from which we see that $T$ has singular values $1$ of multiplicity $(d-1)^2-1$, $e^{-\tfrac{1}{2}}$ of multiplicity $2(d-1)$, $\frac{\sqrt{1-e+e^2+(e-1)\sqrt{1+e^2}}}{e}\approx 1.200$ of multiplicity $1$ and $\frac{\sqrt{1-e+e^2-(e-1)\sqrt{1+e^2}}}{e}\approx 0.306$ of multiplicity $1$. In particular, we have 
\begin{align*}
    \det (T) \leq \left(s^\uparrow_1 (T)\right)^\frac{d}{2}
\end{align*}
but
\begin{align*}
    \det (T) > \left(s^\uparrow_1 (T)\right)^d.
\end{align*}
More precisely, we see that $\det (T) \leq \left(s^\uparrow_1 (T)\right)^p$ requires, as $d\to\infty$,
\begin{align*}
    p
    \leq \frac{\ln (s_1^\downarrow(T)) + \ln(s_1^\uparrow(T)) - (d-1)}{\ln(s_1^\uparrow(T))}
    \approx \frac{\ln (1.200) + \ln(0.306) -(d-1)}{\ln(0.306)}
    \sim \frac{1}{-\ln(0.306)}d \approx 0.845 ~d.
\end{align*}
If we do the same computation for $\frac{1}{n}L$ instead of $L$, we obtain, in the limit of large $n$, the upper bound
\begin{align*}
    p \leq \frac{2}{1+\sqrt{2}}~d + 1 + \frac{\sqrt{2}}{1+\sqrt{2}},
\end{align*}
which coincides up to an additive constant with the bound obtained above on the level of eigenvalues.\\
This concludes our discussion of the example.
\end{exm}

The result of Theorem \ref{CrlQuDetVSPowerSmallestSV} applied to the qubit case does not reproduce the criterion from Theorem \ref{ThmQubitInfinitesimalDivisible}. In particular, we do not obtain $s_{\textrm{min}}^2$ but merely $s_\textrm{min}$. For normal Lindbladians and thus products of unital channels, our reasoning can, however, be improved.

\begin{prp}
For normal Lindbladians the prefactor in Lemma \ref{LmmEigIneqQuantumPowerCriterion} (and thus the exponent in Corollary \ref{CrlQuDetVSPowerSmallestSV}) can be improved to $d$. Furthermore, this estimate is sharp, i.e., cannot be improved for general normal $\L$. 
\end{prp}{}
\begin{proof}
For normal $\L$ we know all the eigenvalues of $L+L^\ast$, they are given by $\{-|\lambda_i-\lambda_j|^2\}_{i,j}$, where $\lambda_i$ are the eigenvalues of $\L$ (see Remark \ref{RmkNormalLindbladians} for a detailed derivation). Now choose two indices $i^\ast,j^\ast$ such that
\[
|\lambda_{i^\ast} -\lambda_{j^\ast}|^2 
=
\max \limits_{i,j} |\lambda_i-\lambda_j|^2.
\]
Then \eqref{eq::EVIneqLindbladGenV1} for exponent $d$ becomes
\begin{align}\label{eq:Exponent_normal_Lindbladian}
    \Tr[L+L^\ast]-d \Lambda_1^\uparrow (L+L^\ast)
    =
    - \sum_{i,j} |\lambda_i - \lambda_j|^2 
    +  d |\lambda_{i^\ast} -\lambda_{j^\ast}|^2 .
\end{align}
Now using $|a+b|^2 \leq 2\big( |a|^2 + |b|^2\big)$ and denoting the indices $\{1,\ldots,d\} \backslash \{i^\ast, j^\ast\} = \{n_1, \ldots, n_{d-2}\}$ we obtain
\begin{align*}
    \eqref{eq:Exponent_normal_Lindbladian}
    \leq
    - \sum_{i,j} |\lambda_i - \lambda_j|^2 
    + 2  |\lambda_{i^\ast} -\lambda_{j^\ast}|^2 
    + 2\sum_{k=1}^{d-2}\left(
    |\lambda_i^\ast - \lambda_{n_k}|^2 + |\lambda_j^\ast - \lambda_{n_k}|^2\right)
    \leq 0.
\end{align*}
In the last step we used that every difference $|\lambda_{i^\ast / j^\ast}-\lambda_{n_k}|^2$ appears twice in the first sum.

In order to see that $d$ is also optimal, consider the example $\L = \mathrm{diag}[1,-1,0,\ldots,0]$.
Here a straightforward calculation shows that $L+L^\ast$ has eigenvalues -4 of multiplicity 2, -1 of multiplicity $4(d-2)$, and 0 of multiplicity $2+(d-2)^2$. Thus, 
\[
\Tr[L+L^\ast] 
= -4d
=
-d |\lambda_1 - \lambda_2|^2 =d\Lambda_1^{\uparrow}(L+L^\ast) ,
\]
so $d$ is optimal.
\end{proof}

Note that the example used in the previous proof can also be used to show that for normal $\L$, the exponent in $\det(e^L)\leq \left(s_1^\uparrow (e^L)\right)^d$ cannot be improved.

\subsubsection{Determinant versus product of smallest singular values}
So far, we have used the ideas from Subsection \ref{SctCriteriaGeneral} to derive an upper bound on the determinant of infinitesimal divisible quantum channels in terms of a power of its smallest singular value. Now we focus on the other aspect of Lemma \ref{LmmSufficientEigAssumptionGeneral} and bound the determinant via a product of smallest singular values.

\begin{lmm}\label{LmmEigIneqQuantumProductCriterion}
Let $L:\mathcal{M}_d\to\mathcal{M}_d$, $L(\rho)=\mathcal{L}\rho\mathcal{L}^\dagger - \frac{1}{2}\{\mathcal{L}^\dagger \mathcal{L},\rho\}$ be a purely dissipative Lindblad generator with one Lindbladian $\mathcal{L}\in\mathcal{M}_d$. Then for $f(d)= 2d-2\sqrt{2d}+1$ we have
\begin{align}
    \Tr[L+L^\ast] - \sum\limits_{K=1}^{\lfloor f(d)\rfloor} \Lambda_K^{\uparrow}(L+L^*) \leq 0.
\end{align}
\end{lmm}
\begin{proof}
As in the proof of Lemma \ref{LmmEigIneqQuantumPowerCriterion}, we can w.l.o.g.~assume $\Tr[\L]=0$ and therefore $\Tr [L + L^\ast]=- 2d\norm{\L}_F^2$.
We can now bound
\begin{align*}
    -\sum\limits_{K=1}^{\lfloor f(d)\rfloor} \Lambda_K^{\uparrow}(L+L^*)
    &\leq \sum\limits_{K=1}^{\lfloor f(d)\rfloor} \lvert \Lambda_K^{\uparrow}(L+L^*)\rvert\\
    &\leq \sum\limits_{K=1}^{\lfloor f(d)\rfloor} s^\downarrow_K(L+L^\ast)\\
    &= \norm{L+L^\ast}_{\left(\lfloor f(d)\rfloor\right)}\\
    &= \norm{\overline{\L} \otimes \L + \overline{\L^\dagger} \otimes \L^\dagger - \mathds{1}_d \otimes \L^\dagger \L - \overline{\L^\dagger \L } \otimes \mathds{1}_d}_{\left(\lfloor f(d)\rfloor\right)}\\
    &\leq 2\norm{\overline{\L} \otimes \L}_{\left(\lfloor f(d)\rfloor\right)} + \norm{\mathds{1}_d \otimes \L^\dagger \L + \overline{\L^\dagger \L } \otimes \mathds{1}_d}_{\left(\lfloor f(d)\rfloor\right)},
\end{align*}
where we used the $k^{\textrm{th}}$ Ky Fan norm of a matrix as 
\begin{align*}
    \norm{A}_{(k)} : = 
    \sum_{i=1}^k s_i^\downarrow (A).
\end{align*}
We bound those two norms separately: For the first term,
\begin{align*}
    \norm{\overline{\L} \otimes \L}_{\left(\lfloor f(d)\rfloor\right)}
    &= \sum\limits_{K=1}^{\lfloor f(d)\rfloor} s^\downarrow_K (\overline{\L} \otimes \L)\\
    &\leq\sqrt{\lfloor f(d)\rfloor} \Big(\sum\limits_{K=1}^{\lfloor f(d)\rfloor} \big(s^\downarrow_K (\overline{\L} \otimes \L)\big)^2\Big)^\frac{1}{2}\\
    &\leq \sqrt{\lfloor f(d)\rfloor} \norm{\overline{\L} \otimes \L}_F\\
    &= \sqrt{\lfloor f(d)\rfloor} \norm{\L}_F^2,
\end{align*}
where the first inequality is an application of Cauchy-Schwarz.\\
For the second term, we choose an ONB w.r.t.~which $\L^\dagger\L$ is diagonal with the squares of the singular values $s_i$ of $\L$ on the diagonal (which is possible by unitary invariance of the Ky Fan norms) and then compute
\begin{align*}
    \norm{\mathds{1}_d \otimes \L^\dagger \L + \overline{\L^\dagger \L } \otimes \mathds{1}_d}_{\left(\lfloor f(d)\rfloor\right)}
    &= \norm{\textrm{diag}[2s_1^2, s_1^2 + s_2^2,\ldots,s_1^2 + s_d^2,s_1^2+s_2^2,\ldots,2s_d^2 ]    }_{\left(\lfloor f(d)\rfloor\right)}\\
    &\leq (\lfloor f(d)\rfloor+1)\sum\limits_{i=1}^d s_i^2\\
    &\leq  (\lfloor f(d)\rfloor+1)\norm{\L}_F^2.
\end{align*}
Plugging this into the above, we obtain 
\begin{align*}
    \Tr[L+L^\ast] - \sum\limits_{K=1}^{\lfloor f(d)\rfloor} \Lambda_K^{\uparrow}(L+L^*)
    &\leq  - 2d\norm{\L}_F^2 + (1 + 2\sqrt{\lfloor f(d)\rfloor} + \lfloor f(d)\rfloor)\norm{\L}_F^2.
\end{align*}
This is $\leq 0$ if $1 + 2\sqrt{f(d)} + f(d)-2d\leq 0$, which is guaranteed by the choice $f(d) = 2d - 2\sqrt{2d} + 1.$
\end{proof}

\begin{rmk}\label{RmkNormalLindbladians}
The reasoning in the proof of Lemma \ref{LmmEigIneqQuantumProductCriterion} becomes particularly simple if the Lindbladian $\L$ is normal. In that case, let $\lbrace v_j\rbrace_j$ be an orthonormal basis for $\IR^d$ consisting of eigenvectors of $\mathcal{L}$ corresponding to eigenvalues $\lbrace \lambda_j\rbrace_j$. By normality, the $\lbrace v_j\rbrace_j$ are also eigenvectors of $\L^\dagger$ to eigenvalues $\lbrace \overline{\lambda}_j\rbrace_j$. Recalling that in the matrix representation we can write $L + L^\ast = \overline{\L} \otimes \L + \overline{\L^\dagger} \otimes \L^\dagger - \mathds{1}_d \otimes \L^\dagger \L - \overline{\L^\dagger \L } \otimes \mathds{1}_d$, it is now easy to see that $\lbrace \bar{v}_i\otimes v_j\rbrace_{i,j}$ is an orthonormal basis of $\IC^{d^2}$ consisting of eigenvectors of $L+L^\ast$ to eigenvalues $\{ -|\lambda_i - \lambda_j|^2\}_{i,j}$. So all eigenvalues of $L+L^\ast$ are $\leq 0$, the inequality of Lemma  \ref{LmmEigIneqQuantumProductCriterion} is trivially satisfied.
\end{rmk}

We can now apply our reasoning from Subsection \ref{SctCriteriaGeneral} (with $k=\lfloor 2d-2\sqrt{2d}+1 \rfloor$ and $p=1$) to obtain

\begin{crl}\label{CrlQuDetVSProductSmallestSV}
Let $T\in\overline{\mathcal{I}}_d$. Then with $f(d)= 2d-2\sqrt{2d}+1$ we have $$0\leq \det (T)\leq\prod\limits_{i=1}^{\lfloor f(d) \rfloor} s_i^{\uparrow}(T).$$
\end{crl}

\begin{exm}\label{ExmOptimalityLinearOrder2}
Consider again the Lindblad generator $L$ from Example \ref{ExmOptimalityLinearOrder} and the corresponding channel $T$. With the eigenvalues and singular values computed in Example \ref{ExmOptimalityLinearOrder}, we see that in this case, $\sum\limits_{i=1}^{d^2-k} \Lambda_i^\downarrow (L+L^\ast)>0$ for all $k\geq 2d - 1$ and we have 
\begin{align*}
    \det (T) \leq \prod\limits_{i=1}^{2d-2} s^\uparrow_i (T)
\end{align*}
but
\begin{align*}
    \det (T) > \prod\limits_{i=1}^{k} s^\uparrow_i (T)
\end{align*}
for every $d^2> k > 2d-2$. This shows that in Corollary \ref{CrlQuDetVSProductSmallestSV} nothing better than $\det (T) \leq \prod\limits_{i=1}^{k} s^\uparrow_i (T)$ with $k=2d-2$ can be achieved.
\end{exm}

\begin{rmk}
After establishing the optimality of picking the smallest $2d-C$ singular values in Corollary \ref{CrlQuDetVSProductSmallestSV}, the question naturally arises whether this bound can in principle be achieved with our proof strategy. In other words, what is the optimal choice for $k$ s.t.~$$\norm{\overline{\L} \otimes \L + \overline{\L^\dagger} \otimes \L^\dagger - \mathds{1}_d \otimes \L^\dagger \L - \overline{\L^\dagger \L } \otimes \mathds{1}_d}_{(k)}\leq 2d \norm{\L}_F^2 ?$$
We clearly have
\begin{align*}
    \norm{\overline{\L} \otimes \L + \overline{\L^\dagger} \otimes \L^\dagger - \mathds{1}_d \otimes \L^\dagger \L - \overline{\L^\dagger \L } \otimes \mathds{1}_d}_{(k)}
    &\leq 2\norm{\overline{\L} \otimes \L}_{(k)} + \norm{\mathds{1}_d \otimes \L^\dagger \L + \overline{\L^\dagger \L } \otimes \mathds{1}_d}_{(k)}.
\end{align*}
The first term has the singular values $s_i(\L)s_j(\L)$ and the second one has singular values $s_i^2(\L)+s_j^2(\L)$. Thus, if we normalize the Frobenius-norm of $\L$ to 1 and write $p_i = s_i^2(\L),$ we can reduce the desired bound to the following 
\begin{cnj}\label{CnjAMGMMatrices}
Let $p\in\IR^d_{\geq 0}$ with $\sum\limits_{i=1}^d p_i = 1$. Define the matrices $a,g\in\IR^{d\times d}$ via
\begin{align*}
    a_{ij} = \frac{p_i + p_j}{2},\quad g_{ij}=\sqrt{p_i p_j}.
\end{align*}
Denote by $a_k^\downarrow$ and $g_k^\downarrow$ the $k^{th}$ largest entry of $a$ and $g$, respectively. Define 
\begin{align*}
    A = \sum\limits_{k=1}^{h(d)} a_k^\downarrow,\quad G = \sum\limits_{k=1}^{h(d)} g_k^\downarrow.
\end{align*}
We conjecture that the maximal integer $h(d)$ such that $A+G\leq d$ holds for any probability vector $p$ is given by $h(d) = 2d-5$.
\end{cnj}
We have tested this conjecture numerically for a wide range of dimensions. Theoretically it stems from the fact that we know the optimal values and corresponding probability vectors for the arithmetic ($h(d) = 2d-2$) and geometric mean ($h(d) = d^2$), respectively. So $A$ is by far more decisive and $G$ can only worsen the maximal number of summands by a bit. 
If we were able to prove this conjecture, we could choose $f(d) = h(d) =2d-5$ in Corollary \ref{CrlQuDetVSProductSmallestSV}, which would bring us closer to the optimum of $2d-2$ up to an additive constant.
\end{rmk}

\begin{rmk}
In contrast to the previous subsection, here we cannot provide an example of a quantum channel that violates the criterion from Corollary \ref{CrlQuDetVSProductSmallestSV}. As any channel having only singular values $\leq 1$ trivially satisfies the criterion, no unital channel will provide a violation, which makes analytically constructing an example more difficult.\\
We have also tried to find an example of a non infinitesimal divisible channel that is recognized as such by the conjectured optimal version of our criterion (which we cannot prove yet) numerically via minimizing the fraction $\prod\limits_{i=1}^{2d-2} s_i^\uparrow (T) / \det (T)$ over channels. This has, however, not been successful. We would be interested in any comments as to how such an example can be found or why finding one is a challenging task.
\end{rmk}

So far in our treatment of infinitesimal divisible quantum channels, we considered two extreme cases,  namely, estimating the determinant by the highest possible power of the smallest singular value and by the product of the largest possible number of the lowest singular values all with exponent $1$. 
The next Proposition corresponds to an interpolation between those two results.
\begin{prp} \label{PropInterpolation}
Let $T\in\overline{\mathcal{I}}_d$. Then for any $1 \leq k \leq d^2$ with $g(d)=\frac{2d}{k+2\sqrt{k}+1}$ we have 
\[
0\leq \det (T)\leq
\left( \prod\limits_{i=1}^{k} s_i^{\uparrow}(T)\right)^{g(d)}.
\]
\end{prp}{}
\begin{proof}
As shown in Subsection \ref{SctCriteriaGeneral}, it suffices to show that any Lindblad generators $L$ satisfies
\[
\Tr[L+L^\ast] - g(d) \sum_{\ell=1}^{k}\Lambda_\ell^{\uparrow}(L+L^\ast)
\leq 
-2d \norm{\L}_{F}^2 + g(d) \norm{L+L^\ast}_{(k)}
\leq 0.
\]
Again, we only need to consider purely dissipative Lindblad generators with a single Lindbladian. For such generators, the desired assertion follows from the bound on the Ky Fan norm provided in the proof of Lemma \ref{LmmEigIneqQuantumProductCriterion}
\[
\norm{L+L^\ast}_{(k)} \leq \big( k + 2\sqrt{k}+1 \big) \norm{\L}_F^2.
\]
\end{proof}{}
\begin{rmk}
In our numerical tests, we observe the result of Corollary \ref{CrlQuDetVSPowerSmallestSV} to be the strongest in generic cases in higher dimensions, since generically the smallest singular value seems to be of some orders of magnitude smaller then the others. But the result in Proposition \ref{PropInterpolation} might give useful improvements for small dimension, especially if some of the lowest singular values are all of the same order of magnitude.
Take, e.g., the case $d=3,k=2$, then we get the three results
\[
0 \leq \det(T) 
\leq
\begin{cases}
s_1^{\uparrow}(T)^{\nicefrac{3}{2}} 
& \textrm{Corollary \ref{CrlQuDetVSPowerSmallestSV}},\\
s_1^{\uparrow}(T)s_2^\uparrow(T) 
& \textrm{Corollary \ref{CrlQuDetVSProductSmallestSV}},\\
\left( s_1^{\uparrow}(T)s_2^\uparrow(T)\right) ^{\frac{6}{3+2\sqrt{2}}} 
& \textrm{Proposition \ref{PropInterpolation}}.
\end{cases}{}
\]
So if $s_1^{\uparrow}(T)$ is a lot smaller than $s_2^{\uparrow}(T)$, the first result is the strongest. But if $s_1^{\uparrow}(T) \approx s_2^{\uparrow}(T)$, then the last result becomes the strongest criterion out of the three.
\end{rmk}{}

\subsection{Stochastic Matrices}\label{SctCriteriaStochastic}
The classical counterparts of quantum channels and Lindblad generators are stochastic matrices and transition rate matrices, respectively. In particular, when choosing the set of generators to be the set of all transition rate matrices, we obtain a notion of (infinitesimal) Markovian divisibility for stochastic matrices.\\
Motivated by the results of Subsections \ref{SctCriteriaGeneral} and \ref{SctCriteriaQuantum} we now study whether similar criteria for infinitesimal divisibility of stochastic matrices can be established. More precisely, we define

\begin{dff}\emph{(Markovian Divisible Stochastic Matrices)}\label{DffMarkovianDivisibleStochastic}\\
We define the set of $d\times d$ stochastic matrices to be 
\[\mathcal{S}_d := \{S\in\IR^{d\times d}~|~ S_{ij}\geq 0~\forall i,j\textrm{ and }\sum\limits_{j=1}^d S_{ij} = 1~\forall i \}\] 
and the set of $d\times d$ transition rate matrices to be 
\[\mathcal{Q}_d:=\{ Q\in\IR^{d\times d}~|~ Q_{ij}\geq 0~\forall i\neq j \textrm{ and }\sum\limits_{j=1}^d Q_{ij}=0~\forall i\}.\]
We call a stochastic matrix $S\in\mathcal{S}_d$ \emph{Markovian divisible} if it is Markovian divisible w.r.t.~the set of generators $\mathcal{Q}_d$ in the sense of Definition \ref{DffMarkovianDivisibleGeneral}.
\end{dff}

Note that, as discussed in Remark \ref{RmkInfinitesimalMarkovianversusMarkovianDivisible}, the ``infinitesimal'' requirement is automatically contained in this definition due to the structure of the set $\mathcal{Q}_d$, which is why we do not write it out explicitly.\\

Our first observation is that, in contrast to the case of Lindblad generators studied in Subsection \ref{SctCriteriaQuantum}, when allowing all transition rate matrices as generators, no non-trivial necessary criteria of our desired form \eqref{eq:GoalInequality} can hold.
\begin{exm}\label{ExmCounterexampleStochastic}
Take the transition rate matrix 
\[Q=
\begin{pmatrix}
-1 & 0 & \ldots & 0 & 1\\
0 & 0 & \ldots & 0 & 0\\
\vdots & & \ddots & & \vdots\\
0 & 0 & \ldots & 0 & 0
\end{pmatrix}\in\mathcal{Q}_d,
\quad \text{then } e^Q = 
\begin{pmatrix}
\frac{1}{e} & 0 & \ldots & 0 & 1-\frac{1}{e}\\
0 & 1 & \ldots & 0 & 0\\
\vdots & & \ddots & & \vdots\\
0 & 0 & \ldots & 0 & 1
\end{pmatrix},\] which has singular values $\frac{\sqrt{1-e+e^2+(e-1)\sqrt{1+e^2}}}{e}\approx 1.200$ of multiplicity $1$, $1$ of multiplicity $d-2$ and $\frac{\sqrt{1-e+e^2-(e-1)\sqrt{1+e^2}}}{e}\approx 0.306$ of multiplicity $1$. In particular, we see that for every $1\leq k<d$
\begin{align*}
    \det(e^Q) > \prod\limits_{i=1}^k s_i^\uparrow (e^Q).
\end{align*}
So for Markovian divisible stochastic matrices, there cannot be a non-trivial necessary criterion of the form of Corollary \ref{CrlQuDetVSProductSmallestSV}. Similarly, no non-trivial necessary criterion as in Corollary \ref{CrlQuDetVSPowerSmallestSV} with an exponent growing with some positive power of $d$ can hold when we take the set $\mathcal{G}$ of generators to be all transition rate matrices.
\end{exm}

This example, together with Corollaries \ref{CrlQuDetVSProductSmallestSV} and \ref{CrlQuDetVSPowerSmallestSV}, implies the following 
\begin{crl}
There cannot be a mapping from $d^2 \times d^2$ stochastic matrices to $\mathcal{T}_d$ that both preserves infinitesimal Markovian divisibility and leaves singular values invariant.
\end{crl}

We can, however, restrict our attention to strict subsets of all transition rate matrices and derive analogous criteria there.
\begin{lmm}
Let $c\in (0,1]$. Consider the set of generators $$\mathcal{G}_c:= \{Q\in\IR^{d\times d}~|~ Q\textrm{ is a transition rate matrix with } \ins{\sum_{\ell=1}^d Q_{\ell k}=0 \textrm{ and }} Q_{kk}\leq c\min\limits_{1\leq l\leq d}Q_{ll}~\forall 1\leq k\leq d\}.$$ Then $\Tr[Q+Q^T] - \frac{1+c(d-1)}{2}\lambda^\uparrow_1 (Q+Q^T)\leq 0$.
\end{lmm}
\begin{proof}
Clearly, for $Q\in\mathcal{G}_c$ we have $\Tr[Q+Q^T]=2\sum\limits_{i=1}^d Q_{ii}\leq 2(1+c(d-1))\min\limits_{1\leq l\leq d}Q_{ll}$. As $\sum\limits_{j=1}^d Q_{ij}=0$ \ins{and $\sum\limits_{j=1}^d Q_{ji}=0$} for all $1\leq i\leq d$, we can use Gerschgorin discs to obtain $\lambda^\uparrow_1 (Q+Q^T)\geq 4\min\limits_{1\leq l\leq d}Q_{ll}$. In particular, we have that 
$$ \Tr[Q+Q^T] - \frac{1+c(d-1)}{2}\lambda^\uparrow_1 (Q+Q^T) \leq 2(1+c(d-1))\min\limits_{1\leq l\leq d}Q_{ll} - 2(1+c(d-1))\min\limits_{1\leq l\leq d}Q_{ll} = 0,$$ as claimed.
\end{proof}

According to our reasoning from Subsection \ref{SctCriteriaGeneral}, this directly implies the following
\begin{crl}\label{CrlStochDetVSPowerSmallestSV}
Let $c\in (0,1]$. Suppose that $S\in [0,1]^{d\times d}$ is a stochastic matrix that is Markovian divisible w.r.t.~$\mathcal{G}_c$. Then $\det (S)\leq \left(s_1^\uparrow (S)\right)^{\frac{1+c(d-1)}{2}}$.
\end{crl}
If we set $c=1$, then $\mathcal{G}_1$ describes the set of transition rate matrices with constant diagonal. For Markovian divisibility of a stochastic matrix $S$ w.r.t.~this restricted set of generators, we obtain again the criterion $\det (S)\leq \left(s_1^\uparrow (S)\right)^{\frac{d}{2}}$.

\section{Conclusion}
In this work, we described how the notion of infinitesimal Markovian divisibility introduced in \cite{Wolf.2008} for quantum channels with the generators in Lindblad form can be extended to a notion applicable to general linear maps and (closed and convex) set of generators.\\
Our main contribution towards an understanding of this notion is a general proof strategy, based on (sub-)multiplicativity properties of the determinant and of products of largest singular values as well as Trotterization, with which we can establish necessary criteria for infinitesimal Markovian divsibility from a spectral property of the generators.\\
We showed that all Lindblad generators satisfy such a property, therefore our approach yields necessary criteria for infinitesimal Markovian divisibility of quantum channels in any (finite) dimension. These are the first such criteria beyond dimension $2$ aside from non-negativity of the determinant. Using these criteria, we gave new examples of provably not infinitesimal Markovian divisible quantum channels that can be found in any neighborhood of any rank-deficient quantum channel.\\
However, when studying the classical counterpart - stochastic matrices as maps of interest and transition rate matrices as generators - we found that in the general scenario in which all possible transition rate matrices are allowed as generators, no necessary criterion of our desired form can hold. We could apply our proof strategy only after imposing an additional restriction on the allowed transition rate matrices, which can be interpreted as requiring that the time scales for remaining in any of the states of the Markov chain are comparable. (In particular, we have to assume that there are no absorbing states.)\\

Several follow-up questions arise naturally from our work. The first such question is for improvements of our results of Corollaries \ref{CrlQuDetVSPowerSmallestSV} and \ref{CrlQuDetVSProductSmallestSV}. In Examples \ref{ExmOptimalityLinearOrder} and \ref{ExmOptimalityLinearOrder2}, we have shown that our results are close to optimal w.r.t.~the dimension dependence of the exponent in Corollary \ref{CrlQuDetVSPowerSmallestSV} and optimal in leading order w.r.t.~the number of factors in Corollary \ref{CrlQuDetVSProductSmallestSV}. Nevertheless, there remains a gap to be closed. One possible step for improving Corollary \ref{CrlQuDetVSProductSmallestSV} might lie in a better understanding of Conjecture \ref{CnjAMGMMatrices}. One might also wonder whether there a subclass of Lindblad operators for which our proof strategy yields stronger bounds.\\
More generally, we are hoping for a better understanding of the result of Corollary \ref{CrlQuDetVSProductSmallestSV}. A crucial first step would be to find - either analytically or numerically - examples of not infinitesimal Markovian divisible quantum channels that violate the inequality in Corollary \ref{CrlQuDetVSProductSmallestSV} (or, for that matter, our conjectured improvement of it). As our proof of this inequality makes extensive use of the assumed divisibility structure, we would consider it surprising if no such examples could be found, which would make it trivial as a necessary criterion.\\
We mention one more natural question concerning the case of infinitesimal Markovian divisible quantum channels. Namely, now that we have established necessary criteria for this property, can these be complemented by sufficient criteria of a similar form? The results of \cite{Wolf.2008} show that for generic qubit channels, an inequality between the determinant of a channel and the square of its smallest singular value is indeed both a necessary and sufficient criterion for infinitesimal Markovian divisibility. But it is not at all clear whether this gerneralizes to higher dimensions.\\
Finally, here we have applied our general proof strategy to two scenarios, that of Lindblad generators and that of transition rate matrices as generators. It would be interesting to find other sets of matrix semigroups whose generators satisfy a spectral property as required in Theorem \ref{ThmNecessaryCondMarkovianDivisibleGeneral}.

\vfill
\section*{Acknowledgements}
Both M.C.C.~and B.R.G.~thank Michael M.~Wolf for suggesting this problem and for many insightful discussions. We also are grateful for the suggestions made by the anonymous reviewer at the Journal of Mathematical Physics.

M.C.C.~gratefully acknowledges support from the TopMath Graduate Center of the TUM Graduate School at the Technical University of Munich, Germany, and from the TopMath Program at the Elite Network of Bavaria. M.C.C.~is supported by a doctoral scholarship of the German Academic Scholarship Foundation (Studienstiftung des deutschen Volkes).

B.R.G.~gratefully acknowledges support from the International Research Training Group
IGDK Munich - Graz funded by the Deutsche Forschungsgemeinschaft (DFG, German Research Foundation) - Projektummer 188264188/GRK1754.
\newpage
\setcounter{secnumdepth}{0}
\defbibheading{head}{\section{References}}
\printbibliography[heading=head]

\setcounter{secnumdepth}{2}
\newpage
\appendix
\section*{Appendix}
\section{Proof of an Improvement to Corollary \ref{CrlQuDetVSPowerSmallestSV}}\label{SctAppendixImprovement}
As mentioned in Remark \ref{remark:improvement}, we are able to improve the exponent in Corollary \ref{CrlQuDetVSPowerSmallestSV} from $\frac{d}{2}$ to $\frac{2}{2+\sqrt{\frac{13}{8}}} ~d \approx 0.610733 ~d$.

The idea behind the improvement is to estimate more carefully the smallest (``most negative'') eigenvalue $ \Lambda_1^\uparrow (L+L^\ast)$. In the proof of Corollary \ref{CrlQuDetVSPowerSmallestSV}, we simply estimate $ \Lambda_1^\uparrow (L+L^\ast)$ from below by $-4\norm{\L}_F^2$, which yields the exponent $\frac{d}{2}$ when comparing it to the $-2d \norm{L}_F^2$ from the trace of $L+L^\ast$.
To obtain our improved version, we prove the following
\begin{lmm} \label{lem:improvement}
Let $\L\in\mathcal{M}_d$ and $L(\rho)=\mathcal{L} \rho\mathcal{L}^{\dagger}- \frac{1}{2} \lbrace \mathcal{L}^{\dagger} \mathcal{L},\rho\rbrace$. Then 
\[
\Lambda_1^\uparrow (L+L^\ast) \geq - \bigg( 2 + \sqrt{\tfrac{13}{8}} \bigg) \norm{\L}_{F}^2.
\]
\end{lmm}{}
\begin{proof}

The starting point for our reasoning is the $l^2$-version of the Gerschgorin disc Theorem (see, e.g., \cite{Bhatia.1997}), which states that for a Hermitian matrix $A = \big( a_{ij}\big)_{i,j}$, each interval $[a_{ii}-r_{i}, a_{ii}+r_i]$ contains at least one eigenvalue of $A$, where
\[
r_i = \left(\sum_{j \neq i} |a_{ij}|^2\right)^{\nicefrac{1}{2}}.
\]
Next, note that due to the tensor-structure of $L+L^\ast$ we can write its entries in a matrix representation as  
\[
\big( L+L^\ast\big)_{kl}
=
\overline{\L}_{(q+1)(p+1)} \L_{rs} + \L_{(p+1)(q+1)} \overline{\L}_{sr} - \delta_{qp} (\L^\dagger \L)_{rs} - (\overline{\L^\dagger\L})_{(q+1)(p+1)} \delta_{rs}
,
\]
where $k = qd+r, l=pd+s$ with $q \in \{0,\ldots,d-1 \}, r \in \{1, \ldots,d \}$.
If we now choose an orthonormal basis such that $\L^\dagger\L= \mathrm{diag}[\sigma_1^2, \ldots,\sigma_d^2]$, we obtain for the diagonal entries
\[
\big( L+L^\ast\big)_{kk}
=
\overline{\L}_{(q+1)(q+1)} \L_{rr} + \L_{(q+1)(q+1)} \overline{\L}_{rr} - \sigma_r^2 -\sigma_{(q+1)}^2.
\]
For the off-diagonal entries we need to consider only the first two terms in $L+L^\ast$ due to the choice of our basis, i.e., we get for $k\neq l$
\[
\big( L+L^\ast\big)_{kl}
=
\overline{\L}_{(q+1)(p+1)} \L_{rs} + \L_{(p+1)(q+1)} \overline{\L}_{sr}.
\]

We need to distinguish two cases. \\
\underline{Case $k=1$:}
Here we have
\[
\big( L+L^\ast\big)_{11}
=
2 |\L_{11}|^2 - 2 \sigma_1^2
\]
and 
\begin{align*}
  \sum_{k \neq 1} \big|\big( L+L^\ast\big)_{1k}\big|^2
&=
\sum_{q,r}
\big| \overline{\L}_{1(q+1)} \L_{1r} + \L_{(q+1) 1} \overline{\L}_{r1} \big|^2 \\
&\leq 
\sum_{q} \big| \overline{\L}_{1(q+1)}\big|^2 \sum_{r} \big| \overline{\L}_{1r}\big|^2 
+
\sum_{q} \big| \overline{\L}_{(q+1)1}\big|^2 \sum_{r} \big| \overline{\L}_{r1}\big|^2 
+
2 \left(\sum_r \underbrace{\big| \L_{1r} \overline{\L_{r1}}\big| }_{\leq \frac{1}{2} \big( |\L_{1r}|^2  + |\L_{r1}|^2\big)} 
\right)^2\\
&\leq 
\norm{\L}_F^2 \big(\norm{\L}_F^2 + |\L_{11} |^2 \big) + \frac{1}{2} \big( \norm{\L}_F^2 + |\L_{11}|^2 \big)^2,
\end{align*}
where in the last step we used that, since we are summing up the first row and column, only the diagonal entry $|\L_{11}|^2$ appears twice and the sum of the remaining squares can be bounded by one Frobenius norm. 

Before we proceed let us note that w.l.o.g.~we can normalize $\norm{\L}_F^2 =1$ to make the following computations more readable. 
Then we obtain by completing the square,
\[
 \sum_{k \neq 1} \big|\big( L+L^\ast\big)_{1k}\big|^2
 \leq 
 1 + |\L_{11} |^2 + \tfrac{1}{2} \big(1 + |\L_{11}|^2 \big)^2
 =
 \big( \sqrt{\tfrac{3}{2} }  + \sqrt{\tfrac{2}{3}} |\L_{11}|^2 \big)^2 - \tfrac{1}{6} |\L_{11}|^4.
\]
Thus,
\[
\big( L+L^\ast\big)_{11}
-
\left( \sum_{k \neq 1} \big|\big( L+L^\ast\big)_{1k}\big|^2\right)^{\nicefrac{1}{2}}
\geq 
2 |\L_{11}|^2 - 2 \sigma_1^2 - \sqrt{\tfrac{3}{2} }  - \sqrt{\tfrac{2}{3}} |\L_{11}|^2 
\geq - \big( 2 + \sqrt{\tfrac{3}{2} }\big).
\]
So in this case we are even able to bound $a_{ii} - r_{i}$ from below by $-(2+\sqrt{\tfrac{3}{2} }) \norm{\L}_F^2$.

\underline{Case $k\neq1$:}
Here, we obtain for the diagonal entries using Young's inequality
\[
\big( L+L^\ast\big)_{kk}
=
\overline{\L}_{(q+1)(q+1)} \L_{rr} + \L_{(q+1)(q+1)} \overline{\L}_{rr} - \sigma_r^2 -\sigma_{(q+1)}^2
\geq 
- 2 \big| \overline{\L}_{(q+1)(q+1)} \L_{rr} \big| - \norm{\L}_F^2.
\]
Note that the two singular values might be the same but can nevertheless be bounded by just one Frobenius norm, which is the important difference to the case $k=1$. 

For the off-diagonal entries we start off in the same way as above
\begin{align*}
    \sum_{l \neq k}
    \big|\big( L+L^\ast\big)_{kl}\big|^2
    &\leq 
    \sum_{(p,s) \neq (q,r)} 
    \big| \overline{\L}_{(q+1)(p+1)} \L_{rs}\big|^2
    + \big| \L_{(p+1)(q+1)} \overline{\L}_{sr}\big|^2
    + 2 \big|  \overline{\L}_{(q+1)(p+1)} \L_{rs} \L_{(p+1)(q+1)} \overline{\L}_{sr}\big|\\
    &=
    \left(\sum_{p} |\L_{(q+1)(p+1)}|^2 \right)
    \left(\sum_{s} |\L_{rs}|^2 \right)
    + 
     \left(\sum_{p} |\L_{(p+1)(q+1)}|^2 \right)
    \left(\sum_{s} |\L_{sr}|^2 \right)\\
    &~~~ + 2 \left(\sum_{p} |\overline{\L}_{(q+1)(p+1)} \L_{(p+1)(q+1)}| \right)
    \left(\sum_{s} |\L_{rs} \overline{\L}_{sr}| \right)
    -4 |\overline{\L}_{(q+1)(q+1)} \L_{rr} |^2\\
    &\leq \norm{\L}_F^2 \left(\norm{L}_F^2 + \min\{|\L_{rr}|^2,|\L_{(q+1)(q+1)}|^2\} \right)-4 |\overline{\L}_{(q+1)(q+1)} \L_{rr} |^2
    \\
    &~~~
    + \frac{1}{2} \left( \norm{L}_F^2 + |\L_{rr}|^2 \right)\left( \norm{L}_F^2 + |\L_{(q+1)(q+1)}|^2 \right).
\end{align*}
Again normalizing $\norm{\L}_F^2 =1$ and denoting $x= \big|\L_{(q+1)(q+1)}\big|, y = \big|\L_{rr}\big| $ gives us
\begin{align*}
\big( L+L^\ast\big)_{kk} 
-
\left( 
\sum_{l \neq k}\big|\big( L+L^\ast\big)_{kl}\big|^2
\right)^{\nicefrac{1}{2}}
&\geq 
-2 xy - 1 
-
 \left( ( 1 + \min\{x^2,y^2\} )
    + \frac{1}{2} (1 + x^2 )(1 + y^2)-4 x^2y^2
    \right)^{\nicefrac{1}{2}}\\
    &=:g(x,y).
\end{align*}{}
Taking the minimum of the function on the right hand side over (the upper half of) the unit disk $x^2 + y^2 \leq 1$ gives us
\[
\big( L+L^\ast\big)_{kk} 
-
\left( 
\sum_{l \neq k}\big|\big( L+L^\ast\big)_{kl}\big|^2
\right)^{\nicefrac{1}{2}}
\geq 
\min_{B_1(0)}g(x,y)= g\big(\tfrac{1}{\sqrt{2}},\tfrac{1}{\sqrt{2}}\big)
=
- 2 - \sqrt{\tfrac{13}{8}}.
\]
As the second case $k\neq 1$ gives us the worse bound, our final estimate is precisely the statement from Lemma \ref{lem:improvement}.

\end{proof}

Again, this has to be compared to $-2d\norm{\mathcal{L}}_F^2$ in the reasoning of the proof of Corollary \ref{CrlQuDetVSPowerSmallestSV}, whereby we obtain the claimed exponent $\frac{2}{2+\sqrt{\frac{13}{8}}} ~d$ (instead of the previous $\frac{d}{2}$).

\end{document}